\documentclass{gen-j-l}%
\usepackage{amsfonts}
\usepackage{amsmath}
\usepackage{amssymb}
\usepackage{graphicx}%
\newtheorem{theorem}{Theorem}[section]
\newtheorem{lemma}[theorem]{Lemma}
\theoremstyle{definition}

\newtheorem{example}[theorem]{Example}

\theoremstyle{remark}
\newtheorem{remark}[theorem]{Remark}
\numberwithin{equation}{section}
\theoremstyle{plain}

\newtheorem{notation}{Notation}

\copyrightinfo{2017}{Carey Caginalp and Gunduz Caginalp}
\begin{document}
\title[Asset Price Volatility and Price Extrema]{ Asset Price Volatility and Price Extrema}
\author{Carey Caginalp}
\address{Economic Science Institute, Chapman Unviersity, 
Orange, CA 92866 and Mathematics Department, University of Pittsburgh, Pittsburgh, PA\ 15260 
and Mathematics Department, University of Pittsburgh, Pittsburgh, PA\ 15260}
\email{carey\_caginalp@alumni.brown.edu}
\urladdr{http://www.pitt.edu/\symbol{126}careycag/}
\author{Gunduz Caginalp}
\address{Mathematics Department, University of Pittsburgh, Pittsburgh, PA\ 15260}
\email{caginalp@pitt.edu}
\urladdr{http://www.pitt.edu/\symbol{126}caginalp/}
\date{\today}
\keywords{Volatility and Price Trend, Modelling Asset Price Dynamics, Price Extrema,
Market Tops and Bottoms}

\begin{abstract}
The relationship between price volatility and a market extremum is examined
using a fundamental economics model of supply and demand. By examining
randomness through a microeconomic setting, we obtain the implications of
randomness in the supply and demand, rather than assuming that price has
randomness on an empirical basis. Within a very general setting the volatility
has an extremum that precedes the extremum of the price. A key issue is that
randomness arises from the supply and demand, and the variance in the
stochastic differential equation governing the logarithm of price must reflect
this. Analogous results are obtained by further assuming that the supply and
demand are dependent on the deviation from fundamental value of the asset.

\end{abstract}
\maketitle

\section{Introduction}

\subsection{Overview}

In financial markets two basic entities are the expected relative price change and
volatility. The latter is defined as the standard deviation of relative price change
in a specified time period. The expected relative price change is, of course,
at the heart of finance, while volatility is central to assessing risk in a
portfolio. Volatility plays a central role in the pricing of options, which
are contracts whereby the owner acquires the right, but not the obligation, to
buy or sell at a particular price within a specified time interval.

In classical finance, it is generally assumed that relative price change is
random, but volatility is essentially constant for a particular asset \cite{BKM}.


\maketitle

In this way, price change and volatility are essentially decoupled in their treatment. 
In particular, the relative price change per unit time $P^{-1}dP/dt=d\log P/dt$
is given by a sum of a deterministic term that expresses the long term
estimate for the growth, together with a stochastic term given by Brownian motion.

Hence, the basic starting point for much of classical finance, particularly
options pricing (see e.g., \cite{BS,WHD}), is the stochastic equation
for $\log P$ as a function of $\omega\in\Omega$ (the sample space) and $t$
given by
\begin{equation}
d\log P=\mu dt+\sigma dW. \label{dpdt}%
\end{equation}
where $W$ is Brownian motion, with $\Delta W:=W\left(  t\right)  -W\left(
t-\Delta t\right)  \sim\mathcal{N}\left(  0,\Delta t\right)  ,$ so $W$ is
normal with variance $\Delta t$, mean $0,$ and independent increments (see
\cite{BI,SH}). While $\mu$ and $\sigma$ are often assumed to be
constant, one can also stipulate deterministic and time dependent or
stochastic $\mu$ and $\sigma.$ The stochastic differential equation above is
short for the integral form (suppressing $\omega$ in notation) for arbitrary
$t_{1}<t_{2}$
\begin{equation}
\log P\left(  t_{2}\right)  -\log P\left(  t_{1}\right)  =\int_{t_{1}}^{t_{2}%
}\mu dt+\int_{t_{1}}^{t_{2}}\sigma dW
\end{equation}
For $\mu,\sigma$ constant, and $\Delta t:=t_{2}-t_{1},$ one can write%
\begin{equation}
\Delta\log P:=\log P\left(  t_{2}\right)  -\log P\left(  t_{1}\right)
=\mu\Delta t+\sigma\Delta W.
\end{equation}

The classical equation (\ref{dpdt}) can be regarded as partly an empirical
model based on observations about volatility of prices. It also expresses the
theoretical construct of infinite arbitrage that eliminates significant distortions
from the expected return of the asset as a consequence of rational comparison
with other assets such as risk free government (i.e., Treasury) bonds. Hence, this equation can
be regarded as a limiting case (as supply and demand approach infinity) of
other equations involving finite supply and demand \cite{CD} (Appendix A). Thus, it does not lend
itself to modification based upon random changes in finite supply and demand. An
examination of the relationship between volatility and price trends, tops and
bottoms requires analysis of the more fundamental equations involving price
change. A suitable framework for analyzing these problems is the asset flow
approach based on supply/demand that have been studied in \cite{CB, CPS, MA1, MC},
 and references therein.

An intriguing question that we address is the following. Suppose there is an
event that is highly favorable for the fundamentals of an asset. There is the
expectation that there will be a peak and a turning point, but no one knows
when that will occur. By observing the volatility of price, can one determine
whether, and when, a peak will occur in the future? In general, our goal is to
delve deeper into the price change mechanism to understand the relationship
between relative price change and volatility.

Our starting point will be the basic supply/demand model of economics (see
e.g., \cite{WG, HG, PS}). We argue that there is always
randomness in supply and demand. However, for a \emph{given} supply and
demand, one cannot expect nearly the same level of randomness in the resulting
price. Indeed, for actively traded equities, there are many market makers
whose living consists of exploiting any price deviations from the optimal
price determined by the supply/demand curves at that moment. While there will
be no shortage of different opinions on the long term prospects of an
investment, each particular change in the supply/demand curve will produce a
clear, repeatable short term change in the price.

Given the broad validity of the Central Limit Theorem, one can expect that the
randomness in supply and demand of an actively traded asset on a given, small
time interval will be normally distributed. Thus, supply and demand can be
regarded as bivariate normally distributed random variables, with a
correlation that will be close to $-1$ since the random factors that increase
demand tend to decrease supply.

In Sections 2 and 3 we explore the implications of this basic price equation
that involves the ratio of demand/supply. By assuming that the supply and
demand are normally distributed with a ratio of means that are characterized
by a maximum, we prove that an maximum in the expectation of the
price is preceded by an extremum in the price volatility. This means that
given a situation in which one expects a market bottom based on fundamentals,
the variance or volatility can be a forecasting tool for the extremum in the
trading price. Furthermore, in pricing options, this approach shows that the
assumption of constant volatility can be improved by understanding the
relationship between the variance in price and the peaks and nadirs of
expected price.

Subsequently, in Section 3, we generalize the dependence on demand/supply in
the basic model, and find that under a broad set of conditions one has
nevertheless the result that the extremum in variance precedes the expected
price extremum.

In Section 4 we introduce the concept of price change that depends on supply
and demand through the fundamental value. The trader motivations are assumed
to be classical in that they depend only on fundamental value; however, the
price equation involves the finiteness of assets, which is a non-classical
concept. Without introducing non-classical concepts such as the dependence of
supply and demand on price trend, we obtain a similar relationship between the
volatility and the expected price.

In Section 5, we prove that within the assumptions of this model and 
generalizations, the peak of the expected log price occurs after the peak 
in volatility.
\subsection{General Supply/Demand model and stochastics}

We write the general price change model in terms of the price, $P,$ the
demand, $\tilde{D}$, and supply, $\tilde{S}$. In particular, the relative
price change is equal to a function of the excess demand, $\left(  \tilde
{D}-\tilde{S}\right)  /\tilde{S}$ (see e.g., \cite{WG}, \cite{HG}). That is,
we have
\begin{equation}
P^{-1}dP/dt=G\left(  \tilde{D}/\tilde{S}\right)  \label{price}%
\end{equation}
where $G:\mathbb{R}^{+}\rightarrow\mathbb{R}$ satisfies $\left(  a\right)  $
$G\left(  1\right)  =0,$ $\left(  b\right)  $ $G^{\prime}\left(  x\right)  >0$
for all $x\in\mathbb{R}^{+}.$ If symmetry between $\tilde{D}$ and $\tilde{S}$
is assumed, then one can also impose $\left(  c\right)  $ $G\left(
1/x\right)  =-G\left(  x\right)  .$ A prototype function with properties
$\left(  a\right)  -\left(  c\right)  $ is given by $G\left(  x\right)
:=x-1/x.$ 

A basic stochastic process based on $\left(  \ref{price}\right)  $ for $\log
P$ is defined by%
\begin{equation}
d\log P\left(  t,\omega\right)  =a\left(  t,\omega\right)  dt+b\left(
t,\omega\right)  dW\left(  t,\omega\right)  \label{abEqn}%
\end{equation}
for some functions $a$ and $b$ in $H_{2}\left[  0,T\right]  $, the space of
stochastic processes with a second moment integrable on $\left[  0,T\right]  $
(see \cite{SH}). The terms $a\left(  t,\omega\right)  $\ and $b\left(
t,\omega\right)  $ can be identified from $G$ and the nature of randomness
that is assumed. In any time interval $\Delta t,$ there is a random term in
$\tilde{D}$ and $\tilde{S}.$ The assumption is that there are a number of
agents who are motivated to place buy orders. The relative fraction is subject
to randomness so that the deterministic demand, $D(t),$ multiplied by
$1+\frac{\sigma}{2}R\left(  t;\omega\right)  $ for some random variable
$R\left(  t;\omega\right)$. Likewise, one has the deterministic supply,
$S\left(  t\right)  ,$ by $1-\frac{\sigma}{2}R\left( t;\omega\right)$. This
yields, for sufficiently small $\sigma$, the approximation
\begin{equation}
\frac{D\left(  t;\omega\right)  }{S\left(  t;\omega\right)  }-1\rightarrow
\frac{D\left(  t\right)  \left\{  1+\frac{\sigma}{2}R\right\}  }{S\left(
t\right)  \left\{  1-\frac{\sigma}{2}R\right\}  }-1\tilde{=}\frac{D\left(
t\right)  }{S\left(  t\right)  }-1+\frac{D\left(  t\right)  }{S\left(
t\right)  }\sigma R,\label{ds}%
\end{equation}
with $\sigma$ being either constant, time dependent or stochastic. We can then
write%
\[
G\left(  \tilde{D}/\tilde{S}\right)  \tilde{=}G\left(  D/S\right)  +G^{\prime
}\left(  D/S\right)  \left(  \sigma\frac{D}{S}R\right)
\]
and thereby identify $a\left(  t;\omega\right)  =G\left(  D/S\right)  $ and
$b\left(  t;\omega\right)  =\sigma\frac{D}{S}G^{\prime}\left(  D/S\right)  .$
Note that we view the randomness as arising only from the $\sigma R$ term, so
we can assume that $D$ and $S$ are deterministic functions of $t$ at this
point. Later on in this paper we consider additional dependence on $D$ and
$S.$ By assuming that the random variable $R$ is normal with variance $\Delta
t$ and $w\left(t+\Delta t\right) - w\left(t\right)$ is independent of 
$w\left(t\right)-w\left(t+\Delta t\right)$, one obtains the stochastic process below (in
which $D\left(  t\right)  $ and $S\left(  t\right)  $ are deterministic). 

By differentiating $\left(  c\right)$, we note
\[
\frac{1}{x}G^{\prime}\left(  \frac{1}{x}\right)  =xG^{\prime}\left(  x\right)
,
\]
and thereby write the stochastic differential equation%
\[
d\log P\left(  t,\omega\right)  =G\left(  D/S\right)  dt+\frac{1}{2}\left\{
\frac{D}{S}G^{\prime}\left(  \frac{D}{S}\right)  +\frac{S}{D}G^{\prime}\left(
\frac{S}{D}\right)  \right\}  dW\left(  t,\omega\right)  .
\]
In particular, for $G\left(  x\right)  :=x-1/x$ one has%
\[
d\log P=\left(  \frac{D}{S}-\frac{S}{D}\right)  dt+\sigma\left\{
\frac{D}{S}+\frac{S}{D}\right\}  dW.
\]

We are interested in the relationship between volatility and market extrema,
and focus on market tops by using the simpler equation for the function
$G\left(  x\right)  :=x-1$ for which $\left(  c\right)  $ holds only
approximately near $D/S=1.$ The equation is then (see Appendix)%

\begin{equation}
d\log P\left(  t,\omega\right)  =\left(  \frac{D\left(  t\right)  }{S\left(
t\right)  }-1\right)  dt+\sigma\left(  t,\omega\right)  \frac{D\left(
t\right)  }{S\left(  t\right)  }dW\left(  t,\omega\right)  .
\end{equation}
For market bottoms, one can obtain similar results (see Appendix).

We will specialize to $\sigma$ deterministic or even constant below. If we
were to assume that the supply and demand have randomness that is not
necessarily the negative of one another, then we can write instead,%
\begin{equation}
\frac{D\left(  1+\sigma_{a}R_{a}\right)  }{S\left(  1-\sigma_{b}R_{b}\right)
}\tilde{=}\left(  1+\sigma_{a}R_{a}+\sigma_{b}R_{b}\right)  \frac{D}{S}-1\ .
\end{equation}
yielding the analogous stochastic process,
\begin{equation}
d\log P\left(  t,\omega\right)  =\left(  \frac{D\left(  t\right)  }{S\left(
t\right)  }-1\right)  dt+\frac{D\left(  t\right)  }{S\left(  t\right)
}\left\{  \sigma_{a}dW_{a}+\sigma_{b}dW_{b}\right\}  .
\end{equation}

\subsection{Derivation of the stochastic equation}

We make precise the ideas above by starting again with (\ref{price})
where $D\left(  t;\omega\right)  $ and $S\left(  t;\omega\right)  $ are random
variables that are anticorrelated bivariate normals with means $\mu_{D}\left(
t\right)  $ and $\mu_{S}\left(  t\right)  $ and both have variance $\sigma
_{1}^{2}.$ We can regard the means as the deterministic part of the supply and
demand at any time $t$, so that with $\Sigma$ as the covariance matrix
\cite{TO}, we write
\begin{equation}
\left(  D\left(  t;\omega\right)  ,S\left(  t;\omega\right)  \right)
\sim\mathcal{N}\left(  \vec{\mu}\left(  t\right)  ,\Sigma\right)
\ \ with\ \ \vec{\mu}:=\left(  \mu_{D},\mu_{S}\right)  ,\ \Sigma:=\left(
\begin{array}
[c]{cc}%
\sigma_{1}^{2}\left(  t\right)  & -1\\
-1 & \sigma_{1}^{2}\left(  t\right)
\end{array}
\right)  .
\end{equation}
For any fixed $t$, one can show that the density of $D/S$ is given by
\begin{equation}
f_{D/S}\left(  x\right)  =\frac{1+\mu_{D}/\mu_{S}}{\sqrt{2\pi}\frac{\sigma
_{1}}{\mu_{S}}\left(  x+1\right)  ^{2}}e^{-\frac{1}{2}\frac{\left(  x-\mu
_{D}/\mu_{S}\right)  ^{2}}{\left(  \frac{\sigma_{1}}{\mu_{S}}\right)
^{2}\left(  x+1\right)  ^{2}}}.
\end{equation}
Other approximations in different settings have been studied in \cite{DR, H1, H2} and references therein.

For values of $x$ near the mean of $D/S$, one has%
\begin{equation}
\left(  x+1\right)  ^{2}\tilde{=}\left(  \frac{\mu_{D}}{\mu_{S}}+1\right)  ^{2}.
\end{equation}
We can use this to approximate the density, using $\ \sigma_{R_{q}}^{2}%
:=\frac{\sigma_{1}^{2}}{\mu_{S}^{2}}\left(  \frac{\mu_{D}}{\mu_{S}}+1\right)
^{2}$ as the approximate variance of $D/S,$ as%
\begin{equation}
f_{D/S}\left(  x\right)  \tilde{=}\frac{1}{\sqrt{2\pi}\sigma_{R_{q}}}%
e^{-\frac{\left(  x-\mu_{D}/\mu_{S}\right)  ^{2}}{2\sigma_{R_{q}}^{2}}%
};\ \ \ f_{\frac{D}{S}-1}\left(  x\right)  \tilde{=}\frac{1}{\sqrt{2\pi}%
\sigma_{R_{q}}}e^{-\frac{\left(  x-\mu_{D}/\mu_{S}+1\right)  ^{2}}{2\sigma_{R_{q}}%
^{2}}}.
\end{equation}
With this expression for the density of $R_{1}:=D/S-1,$ we can write the basic
supply/demand price change equation as%
\begin{equation}
\frac{\Delta\log P}{\Delta t}\tilde{=}R_{1}\sim\mathcal{N}\left(  \frac
{\mu_{D}}{\mu_{S}}-1,\sigma_{R_{q}}^{2}\right)  ,
\end{equation}
where each variable depends on $t$ and $\omega.$ Subtracting out the
$\frac{\mu_{D}}{\mu_{S}}-1$, defining $R_{0}\sim\mathcal{N}\left(
0,\sigma_{R_{q}}^{2}\right)  ,$ and noting that $R_{0}$ depends on $t$ through
$\sigma_{R}^{2},$ we write%
\begin{equation}
\Delta\log P\tilde{=}\left(  \frac{\mu_{D}}{\mu_{S}}-1\right)  \Delta
t+\sigma_{R}R_{0}\Delta t.
\end{equation}
By definition of Brownian motion, we can write%
\begin{equation}
\Delta\log P\tilde{=}\left(  \frac{\mu_{D}}{\mu_{S}}-1\right)  \Delta
t+\sigma_{R_{q}}\Delta W.
\end{equation}

With $\sigma_{1}$ and $\mu_{D}$ held constant, an increase in $\mu_{S}$ leads
to a decrease in the variance $\sigma_{R}.$ We would like to approximate this
under the condition that $\mu_{D}/\mu_{S}\approx1.$ By rescaling the units of
$\mu_{D}$, $\mu_{S},\sigma_{1}$ together and assuming that each of $\mu_{D}$
and $\mu_{S}$ are sufficiently close to $1$ that we can consider the leading
terms in a Taylor expansion, and write
\begin{equation}
\mu_{D}=1+\delta_{D},\text{ \ }\mu_{S}=1+\delta_{S}\ .
\end{equation}
Note that $\mu_{D}$ and $\mu_{S}$ will be nearly equal unless one is far from
equilibirium. Ignoring the terms higher than first order one has%
\begin{align}
\sigma_{R_{q}}^{2}  &  =\frac{\sigma_{1}^{2}}{\left(  1+\delta_{S}\right)  ^{2}%
}\left(  1+\frac{1+\delta_{D}}{1+\delta_{S}}\right)  ^{2}\nonumber\\
&  \tilde{=}4\sigma_{1}^{2}\left(  1-3\delta_{S}+\delta_{D}\right)  .
\end{align}
We are considering $-\delta_{S}=\delta_{D}=:\delta$ so that
\begin{equation}
\sigma_{R_{q}}^{2}=4\sigma_{1}^{2}\left(  1+4\delta\right)  .
\end{equation}
Using Taylor series approximation, one has
\begin{equation}
\left(  \frac{\mu_{D}}{\mu_{S}}\right)  ^{2}=\left(  \frac{1+\delta}{1-\delta
}\right)  ^{2}\tilde{=}1+4\delta.
\end{equation}
We can thus write the stochastic equation above as%
\begin{equation}
\Delta\log P\tilde{=}\left(  \frac{\mu_{D}}{\mu_{S}}-1\right)  \Delta
t+2\sigma_{1}\frac{\mu_{D}}{\mu_{S}}\Delta W,
\end{equation}
so that the differential form is given in terms of $f:=\mu_{D}/\mu_{S}-1$ by%
\begin{equation}
d\log P\left(  t\right)  =f\left(  t\right)  dt+\sigma\left(  f\left(
t\right)  +1\right)  dW\left(  t\right)  \label{stoch}%
\end{equation}
This is in agreement with the heuristic derivation above, with $\sigma
=2\sigma_{1}$ and $\sigma_{1}^{2}$ \ as the variance of each of $S$ and $D.$

\section{Location of maxima of Supply/Demand versus price}

\subsection{The deterministic model.}

We will show that if $D/S-1$ is given by a deterministic function $f$, then
the stochastic equation above will imply that the variance over a small time
interval $\Delta t$ will have an extremum before the price has its extremum.

Once we do this simplest case, it will generalize it to the situation where
$f:=D/S-1$ is also stochastic, and show that the same result holds.

To this end, first consider the simple, purely deterministic case:
\begin{equation}
P^{-1}\frac{dP}{dt}=\frac{D}{S}-1=:f,\ i.e.,\ \ \frac{d}{dt}\log P\left(
t\right)  =f\left(  t\right)  \label{stochasticEquation}%
\end{equation}

Assume that $f$ is a prescribed function of $t$ that is $C^{1}\left(
I\right)  $ for $I\supset$ $\left(  t_{0},\infty\right)  \supset\left(
t_{a},t_{b}\right)  $ satisfying:

$\left(  i\right)  $ $f\left(  t\right)  >0$ on $\left(  t_{a},t_{b}\right)
,\ f\left(  t\right)  <0$ on $I~\backslash\ \left(  t_{a},t_{b}\right)  $ and
$f+1>0$ on $I;$

$\left(  ii\right)  $ $f^{\prime}\left(  t\right)  >0$ if $t\,<t_{m}\ ,$
$f^{\prime}\left(  t\right)  <0$ if $t\,>t_{m}\ ,$ $f^{\prime}\left(
t_{m}\right)  =0;$

$\left(  iii\right)  $ $f^{\prime\prime}\left(  t\right)  <0$ if $t\in\left(
t_{a},t_{b}\right)  .$

Then $\log P\left(  t\right)  $ is increasing on $t\in\left(  t_{a}%
,t_{b}\right)  $ and decreasing on $t\in\left(  t_{b},\infty\right)  $ and has
a maximum at $t_{b}.$

In other words, the peak of $f$ occurs at $t_{m}$ while the peak of $\log P$
is attained at $t_{b}>t_{m}.$ This demonstrates the simple idea that price
peaks some time after the peak in demand/supply. In fact, during pioneering
experiments Smith, Suchanek and Williams \cite{SSW} observed that bids tend to
dry up shortly before a market peak. Also, the important role of the ratio of cash
to asset value in a market bubble that was predicted in \cite{CB} was
confirmed in experiments starting with \cite{CPS}.

\subsection{The stochastic model.}

Recall that $\mu_{D}$ and $\mu_{S}$ are deterministic functions of time only.
We model the problem as discussed above so the only randomness below is in the
$dW$ variable. The stochastic equation given by $\left(  \ref{stoch}\right)  $
for a continuous function $f:=\mu_{D}/\mu_{S}-1,$ in the integral form, for
any $t_{1}<t_{2}$ and $\Delta\log P:=\log P\left(  t_{2}\right)  -\log
P\left(  t_{1}\right)  $ is
\begin{equation}
\Delta\log P=\int_{t_{1}}^{t_{2}}f\left(  z\right)  dz+\int_{t_{1}}^{t_{2}%
}\sigma\left(  z\right)  \left(  f\left(  z\right)  +1\right)  dW\left(
z\right)  .
\end{equation}
Note that for the time being we are assuming that $\sigma$ and $f$ may depend
on time but are deterministic. We compute the expectation
\footnote{We let $\mathbb{E}\left[Y\right]^{2}$ denote 
$\mathbb{E}\left[\left(Y^{2}\right)\right]$.} and variance of this
quantity:%
\begin{equation}
E\left[  \Delta\log P\right]  =\int_{t_{1}}^{t_{2}}f\left(  z\right)  dz
\end{equation}
since $f$ is deterministic and $E\left[  dW\right]  =0;$

\begin{align}
Var\left[  \Delta\log P\right]   &  =E\left[  \int_{t_{1}}^{t_{2}}f\left(
z\right)  dz+\int_{t_{1}}^{t_{2}}\sigma\left(  z\right)  \left\{  f\left(
z\right)  +1\right\}  dW\left(  z\right)  \right]  ^{2}\nonumber\\
&  -\left(  E\left[  \int_{t_{1}}^{t_{2}}f\left(  z\right)  dz+\int_{t_{1}%
}^{t_{2}}\sigma\left(  z\right)  \left\{  f\left(  z\right)  +1\right\}
dW\left(  z\right)  \right]  \right)  ^{2}.\label{varLogP}
\end{align}
The $\int f\left(  z\right)  dz$ term is deterministic and vanishes when its
expectation is subtracted. The expecation of the $dW$ and the $dzdW$ terms
vanishes also. We are left with%
\begin{align}
Var\left[  \Delta\log P\right]   &  =E\left[  \int_{t_{1}}^{t_{2}}%
\sigma\left(  z\right)  \left\{  f\left(  z\right)  +1\right\}  dW\left(
z\right)  \right]  ^{2}\nonumber\\
&  =\int_{t_{1}}^{t_{2}}\sigma^{2}\left(  z\right)  \left\{  f\left(
z\right)  +1\right\}  ^{2}dz
\end{align}
using the standard result (\cite{SH}, p. 68).

We want to consider a small interval $\left(  t,t+\Delta t\right)  $ so we set
$t_{1}\rightarrow t$ and $t_{2}\rightarrow t+\Delta t$. We have%
\begin{align}
V\left(  t,t+\Delta t\right)   &  :=Var\left[  \log P\left(t+\Delta
t\right)  -\log P\left(  t\right)  \right] \nonumber\\
&  =\int_{t}^{t+\Delta t}\sigma^{2}\left(  z\right)  \left\{  f\left(
z\right)  +1\right\}  ^{2}dz.\\
\mathbb{V}\left(  t\right)   &  :=\lim_{\Delta t\rightarrow0}\frac{1}{\Delta
t}V\left(  t,t+\Delta t\right)  =\lim_{\Delta t\rightarrow0}\frac{1}{\Delta
t}\int_{t}^{t+\Delta t}\sigma^{2}\left(  z\right)  \left\{  f\left(  z\right)
+1\right\}  ^{2}dz\nonumber\\
&  =\sigma^{2}\left(  t\right)  \left\{  f\left(  t\right)  +1\right\}  ^{2}.
\label{boldVDef}%
\end{align}

\begin{example}
For $\sigma:=1,$ the maximum variance of $\Delta\log P$ will be when $\left\{
f\left(  z\right)  +1\right\}  ^{2}$ is at a maximum, which is when $f$ has
its maximum, i.e., at $t_{m}\ .$
\begin{equation}
\frac{d}{dt}\mathbb{V}\left(  t\right)  =\frac{d}{dt}\left\{  f\left(
t\right)  +1\right\}  ^{2}=2\left\{  f\left(  t\right)  +1\right\}  \frac
{d}{dt}f\left(  t\right)\label{dv}
\end{equation}
Since $1+f\left(  t\right)  >0$ in all cases, we see that the derivative of
$\mathbb{V}\left(  t\right)  $\ is of the same sign as the derivative of $f,$
so the limiting variance $\mathbb{V}\left(  t\right)  $ is increasing when $f$
is increasing and vice-versa. Recall that $\log P$ increases so long as $f>0,$
and decreases when $f<0.$ In other words, for the peak case, one has $f\left(
t\right)  >0$ if and only if $t\in\left(  t_{a},t_{b}\right)  $ with a maximum
at $t_{m}.$\ When $f$ has a peak, the maximum of $\ \mathbb{V}\left(
t\right)  $ will be at $t_{m}$ when $f\left(  t\right)  $ has its maximum.
\end{example}

To summarize, if the coefficient of $dW$ is $\sigma\left\{  1+f\left(
t\right)  \right\}  $ with $\sigma$ constant and $f$ has a maximum at
$t_{m}$ then $\mathbb{V}\left(  t\right)  $ will also have a maximum 
at $t_{m}$ so that the maximum in $E\log P$ will occur
after the maximum in $\mathbb{V}\left(  t\right)  $ since $\partial_{t}E\log
P\left(  t\right)  =f\left(  t\right).$

\begin{remark}
We have shown that $E\log P\left( t\right) $ has a
maximum, at some time $t_{m}$ that is preceded by a maximum in $\mathbb{V}%
\left( t\right) $. We can use this together with Jensen's inequality to show
that $E\left[ P\left( t_{m}\right) /P\left( t\right) \right] \geq 1$ for
arbitrary $t.$ Indeed, since $E\log P\left( t_{m}\right) \geq E\log P\left(
t\right) $ we can write%
\begin{equation}
E\log \frac{P\left( t_{m}\right) }{P\left( t_{1}\right) }\geq 0.
\end{equation}
Let $Y:=P\left( t_{m}\right) /P\left( t_{1}\right) $ and $g\left( x\right)
:=e^{x}$ in Jensen's inequality, $Eg\left( Y\right) \geq g\left( E\left[ Y%
\right] \right) $, we have%
\begin{equation}
EY=Ee^{\log Y}\geq e^{E\log Y}\geq 1.
\end{equation}
Hence, the expected ratio of price at $t_{m}$ to the price at 
any other point $t$ is greater than $1.$ 
\end{remark}

\begin{remark}
The conclusion above can be contrasted with the standard model 
$\left(\ref{dpdt}\right)$ adjusted so that $\mu\left(t\right):=\frac{\mu_{D}\left(t\right)}
{\mu_{s}\left(t\right)}$ has the same property of a peak at some 
time $t_{m}$. Performing the same calculation of 
$\left(\ref{varLogP}\right)$-$\left(\ref{dv}\right)$ 
for this model yields the result $\mathbb{V}\left(t\right)=\sigma^{2}$ 
so that it provides no information on the expected peak of prices.
\end{remark}

\section{Additional randomness In Supply and Demand}
\subsection{Stochastic Supply and Demand.}

Let $f:=D/S-1$ be a stochastic function such that $-1\leq Ef$ and $E\left\vert
f\right\vert \leq C_{1}$. With $X\left(  t\right)  :=\log P\left(  t\right)  $
and $\Delta X:=X\left(  t+\Delta t\right)  -X\left(  t\right)  ,$ we write the
SDE in differential and integral forms as%
\begin{equation}
dX=fdt+\sigma\left(  1+f\right)  dW\label{1a1}%
\end{equation}%
\begin{equation}
X\left(  t+\Delta t\right)  -X\left(  t\right)  =\int_{t}^{t+\Delta t}f\left(
s\right)  ds+\int_{t}^{t+\Delta t}\sigma\left(  s\right)  \left(  1+f\left(
s\right)  \right)  dW\left(  s\right)  .\label{1b1}%
\end{equation}
where we will assume $\sigma$ is a continuous, deterministic function of time,
though we can allow it to be stochastic in most of the sequel.

One has since $EdW=0$ and $E\left[  dsdW\right]  =0$ one obtains again the identities%
\begin{equation}
E\Delta X=\int_{t}^{t+\Delta t}Ef\left(  s\right)  ds,
\end{equation}%
\begin{align}
Var\left[  \Delta X\right]   &  =E\left[  \int fds+\int\sigma\left(
1+f\right)  dW\right]  ^{2}\nonumber\\
&  -\left(  E\left[  \int fds+\int\sigma\left(  1+f\right)  dW\right]
\right)  ^{2}\nonumber\\
&  =Var\left[  \int fds\right]  +2E\left[  \int fds\int\sigma\left(
1+f\right)  dW\right]  +E\left[  \int\sigma\left(  1+f\right)  dW\right]  ^{2}
\label{1a2}%
\end{align}
where all integrals are taken over the limits $t$ and $t+\Delta t$.

\begin{lemma}
Let $\sup_{\left[  0,T\right]  }E\left\vert f\right\vert ^{2}\leq C^{2}.$ Then
for some $C$ depending on this bound, one has
\begin{equation}
\left\vert E\int_{t}^{t+\Delta t}f\left(  s^{\prime}\right)  ds^{\prime}%
\int_{t}^{t+\Delta t}\sigma\left(  s\right)  \left\{  1+f\left(  s\right)
\right\}  dW\left(  s\right)  \right\vert \leq C\left(  \Delta t\right)
^{3/2}.
\end{equation}

\end{lemma}

\begin{proof}
We apply the Schwarz inequality to obtain
\begin{align}
&  \left\vert E\int_{t}^{t+\Delta t}f\left(  s^{\prime}\right)  ds^{\prime
}\int_{t}^{t+\Delta t}\sigma\left(  s\right)  \left\{  1+f\left(  s\right)
\right\}  dW\left(  s\right)  \right\vert \label{*}\nonumber\\
&  \leq\left\{  E\left(  \int_{t}^{t+\Delta t}f\left(  s^{\prime}\right)
ds^{\prime}\right)  ^{2}\right\}  ^{1/2}\left\{  E\left(  \int_{t}^{t+\Delta
t}\sigma\left(  s\right)  \left\{  1+f\left(  s\right)  \right\}  dW\left(
s\right)  \right)  ^{2}\right\}  ^{1/2}.
\end{align}
We bound each of these terms. Using the Schwarz inequality on the $\int ds$
integral, we obtain using generic $C$ throughout,
\begin{equation}
E\left(  \int_{t}^{t+\Delta t}f\left(  s^{\prime}\right)  ds^{\prime}\right)
^{2}\leq C\left(  \Delta t\right)  ^{2}.\label{bound1}%
\end{equation}
The second term is bounded using the fact that $\sigma$ is deterministic,
\begin{align}
E\left(  \int_{t}^{t+\Delta t}\sigma\left(  s\right)  \left\{  1+f\left(
s\right)  \right\}  dW\left(  s\right)  \right)  ^{2} &  =\int_{t}^{t+\Delta
t}\sigma^{2}\left(  s\right)  E\left\{  1+f\left(  s\right)  \right\}
^{2}ds\nonumber\\
&  \leq C\Delta t.\label{bound2}%
\end{align}
Taking the square roots of (\ref{bound1}) and (\ref{bound2}), and combining
with $\left(  \ref{*}\right)  $ proves the lemma.
\end{proof}

\begin{lemma}
Let $\sigma$ be a continuous, deterministic function and assume $\sup_{\left[
0,T\right]  }E\left\vert f\right\vert ^{2}\leq C^{2}.$ Then%
\begin{equation}
\left\vert Var\left[  \Delta X\right]  -\int_{t}^{t+\Delta t}\sigma^{2}\left(
s\right)  E\left\{  1+f\left(  s\right)  \right\}  ^{2}ds\right\vert \leq
C\left(  \Delta t\right)  ^{3/2} \label{1a3}%
\end{equation}

\end{lemma}

\begin{proof}
Basic stochastic analysis yields%
\begin{equation}
E\left(  \int_{t}^{t+\Delta t}\sigma^{2}\left(  s\right)  \left\{  1+f\left(
s\right)  \right\}  dW\right)  ^{2}=\int_{t}^{t+\Delta t}\sigma^{2}\left(
s\right)  E\left\{  1+f\left(  s\right)  \right\}  ^{2}ds.
\end{equation}
Thus, using (\ref{1a2}) and $f\in H_{2}\left[  0,T\right]  $ we have the
result (\ref{1a3}).
\end{proof}

\bigskip

Now, we would like to determine the maximum of $\mathbb{V}\left(  t\right)  $
and show that it precedes the maximum of the expected log price.$\ $From the calculations
above, one has

\begin{lemma}
\label{Lem extremumv}In the general case, assuming $E\left\vert f\left(
t\right)  \right\vert ^{2}<C^{2}$ on $t\in\left[  0,T\right]  $ but allowing
stochastic $\sigma$ such that $E\sigma^{2}<C$ one has
\begin{equation}
\mathbb{V}\left(  t\right)  :=\lim_{\Delta t\rightarrow0}\frac{1}{\Delta
t}V\left(  t,t+\Delta t\right)  =E\left[  \sigma^{2}\left(  t\right)  \left(
1+f\left(  t\right)  \right)  ^{2}\right]  .
\end{equation}

\end{lemma}

\begin{lemma}
Suppose $\sup_{[0,T]}E\left\vert f\left(  t\right)  \right\vert ^{2}<C^{2}$
$\ $and $\sigma$ is a deterministic continuous function on $t\in\left[
0,T\right]  $ then one has%
\begin{equation}
\mathbb{V}\left(  t\right)  =\sigma^{2}\left\{  1+Ef\right\}  ^{2}+\sigma
^{2}Varf.
\end{equation}
and the extrema of $\ \mathbb{V}\left(  t\right)  $ occur at $t$ such that
\begin{equation}
2\sigma\sigma^{\prime}\left\{  \left[  1+Ef\right]  ^{2}+Varf\right\}
+\sigma^{2}\left\{  2\left[  1+Ef\right]  \left(  Ef\right)  ^{\prime}+\left(
Varf\right)  ^{\prime}\right\}  =0.
\end{equation}

\end{lemma}

\begin{proof}
Using Lemma \ref{Lem extremumv}, we write
\begin{align}
\mathbb{V}\left(  t\right)   &  =\sigma^{2}E\left[  1+2f+f^{2}\right]
=\sigma^{2}\left\{  1+2Ef+\left(  Ef\right)  ^{2}+Ef^{2}-\left(  Ef\right)
^{2}\right\} \nonumber\\
&  =\sigma^{2}\left(  1+Ef\right)  ^{2}+\sigma^{2}Varf.
\end{align}
Differentiation implies the second assertion.
\end{proof}

\bigskip

\begin{lemma}
Suppose $E\left\vert f\left(  t\right)  \right\vert ^{2}<C^{2}$ on
$t\in\left[  0,T\right]  $, while $\sigma$ and $Var\left[  f\left(  t\right)
\right]  $ are constant in $t.$ Then the extremum of $\ \mathbb{V}\left(
t\right)  $ occur for $t$ such that%
\begin{equation}
\frac{d}{dt}Ef\left(  t\right)  =0.
\end{equation}

\end{lemma}

\begin{proof}
From the previous Lemma, we have $V\left(  t,t+\Delta t\right)  :=\int
_{t}^{t+\Delta t}\sigma^{2}\left(  s\right)  E\left[  1+f\left(  s\right)
\right]  ^{2}ds$, yielding%
\begin{equation}
\lim_{\Delta t\rightarrow0}\frac{1}{\Delta t}V\left(  t,t+\Delta t\right)
=\sigma^{2}\left(  1+Ef\left(  t\right)  \right)  ^{2}+Var\left[  f\left(
t\right)  \right]
\end{equation}
Since we are assuming that $Var\left[  f\left(  t\right)  \right]  $ is
constant in time, we obtain%
\begin{align}
\frac{\partial}{\partial t}\lim_{\Delta t\rightarrow0}V\left(  t,t+\Delta
t\right)   &  =\frac{\partial}{\partial t}\left\{  \sigma^{2}\left(
1+Ef\left(  t\right)  \right)  ^{2}\right\} \nonumber\\
&  =2\sigma^{2}\left(  1+Ef\left(  t\right)  \right)  \frac{d}{dt}Ef\left(
t\right)  .
\end{align}
Thus, the right-hand side vanishes if and only if $\frac{d}{dt}Ef\left(  t\right)  =0,$ i.e., at
$t_{m}$ (by definition of $t_{m}$). Note that we have $1+f>0$ so that
$1+Ef>0.$
\end{proof}

\subsection{Properties of $f$}

The condition $E\left\vert f\right\vert ^{2}<C^{2}$ is easily satisfied by
introducing randomness in many forms. For the Lemma above, we would also like
to satisfy $Var\left[  f\left(  t\right)  \right]  =const.$

Another way of attaining this (up to exponential order) is to define $f$ as
the stochastic process%
\begin{equation}
df\left(  t\right)  =\mu_{f}\left(  t\right)  dt+\sigma_{f}\left(  t\right)
dW\left(  t\right)
\end{equation}
where $\mu_{f}$ and $\sigma$ are both time dependent but deterministic.

We can assume that $f\left(  t_{0}\right)  $ is a given, fixed value, and
obtain (see e.g., \cite{BI}, \cite{SH})
\begin{equation}
Var\left[  f\left(  t\right)  \right]  =E\left[  \int_{t_{0}}^{t}\sigma
_{f}\left(  s\right)  dW\left(  s\right)  \right]  ^{2}=\int_{t_{0}}^{t}%
\sigma_{f}^{2}\left(  s\right)  ds
\end{equation}
since $\sigma_{f}\left(  s\right)  $ is deterministic$.$

In particular, if one has $\sigma_{f}\left(  s\right)  :=e^{-s/2}$, then
$Var\left[  f\left(  t\right)  \right]  \leq e^{-t_{0}}$ while $\int
_{0}^{t_{0}}\sigma\left(  s\right)  ds=1-e^{-t_{0}}$ so one has approximately
constant variance for $t\geq t_{0}$ for large $t_{0}$. In particular, one has%
\begin{equation}
\frac{d}{dt}Var\left[  f\left(  t\right)  \right]  =\frac{d}{dt}\int_{t_{0}%
}^{t}\sigma_{f}^{2}\left(  s\right)  ds=\sigma_{f}^{2}\left(  t\right)
=e^{-t}.
\end{equation}

\subsection{General coefficient of $dW$}

The stochastic differential equation (\ref{1a1}) entails a coefficient of $dW$
that is proportional to $D/S.$ One can also consider the implications of a
coefficient that is proportional to the excess demand $D/S-1$ or a monomial of
it. More generally, we can write $h\left(  t\right)  :=g\left(  f\left(
t\right)  \right)  $ for an arbitrary continuous function $g$ leading to the
stochastic differential equation
\begin{equation}
d\log P=fdt+\sigma hdW,
\end{equation}
where $\sigma$ can also be stochastic or deterministic function of time.

From this stochastic equation one has immediately%
\begin{equation}
\frac{dE\left[  \log P\right]  }{dt}=Ef
\end{equation}
similar to the completely deterministic model, except that $f$ is replaced by
$Ef.$

From the integral version of the stochastic model, we can write the
expectation and variance as%
\begin{equation}
E\left[  \Delta\log P\right]  =\int_{t}^{t+\Delta t}Ef\left(  s\right)  ds
\end{equation}%
\begin{align}
V\left(  t,t+\Delta t\right)   &  :=Var\left[  \Delta\log P\right]
=Var\left[  \int_{t}^{t+\Delta t}f\left(  s\right)  ds\right]  +2\mathbb{E}%
\left[  \int_{t}^{t+\Delta t}\sigma\left(  s\right)  h\left(  s\right)
dW\left(  s\right)  \right] \nonumber\\
&  +\int_{t}^{t+\Delta t}E\left[  \sigma\left(  s\right)  h\left(  s\right)
\right]  ^{2}ds.
\end{align}
The middle term on the right-hand side vanishes while the first 
term is of order $\left(\Delta t\right)^{2}$, yielding the following relation for $\mathbb{V}\left(
t\right)$.

\begin{lemma}
Let $h\left(  t\right)  :=g\left(  f\left(  t\right)  \right)  \ and$ $\sigma$
satisfy $Eh^{2}<C,$ $E\sigma^{2}<C$. Then one has
\begin{equation}
\mathbb{V}\left(  t\right)  :=\lim_{\Delta t\rightarrow0}\frac{1}{\Delta
t}V\left(  t,t+\Delta t\right)  =E\left[  \sigma\left(  t\right)  h\left(
t\right)  \right]  ^{2}. \label{E}%
\end{equation}

\end{lemma}

Next, we examine whether $\mathbb{V}\left(  t\right)  $ occurs prior to the
maximum of $\log P\left(  t\right)$ in several examples.

\begin{example}
Consider the function $g\left(  z\right)  =z^{q}$ where $q\in\mathbb{N}$. Let
$\sigma:=1$ and $f\in L^{2}\left[  0,t\right]  $ be deterministic. From the
Lemma above, we obtain%
\begin{equation}
\mathbb{V}\left(  t\right)  =h\left(  t\right)  ^{2}=f\left(  t\right)
^{2q},\ \ \frac{d}{dt}\mathbb{V}\left(  t\right)  =2qf\left(  t\right)
^{2q-1}\frac{d}{dt}f\left(  t\right)  .
\end {equation}

When $f:=D/S-1$ has a maximum, note that on some interval
$\left(  t_{a},t_{b}\right)  $ it is positive (as demand exceeds supply) and
$f$ has its maximum for some value $t_{m}\in\left(  t_{a},t_{b}\right)  .$ The
identity above implies that $\mathbb{V}\left(  t\right)  $ has a maximum when
$f$ has a maximum. Also, the defining stochastic equation above implies $E\log
P$ has its maximum at $t_{b}>t_{m}.$
\end{example}

\begin{example}
(Symmetry between $D$ and $S$ and more general coefficients) If we hypothesize
that the level of noise is proportional essentially to the magnitude (or its
square) of the difference between $D$ and $S$ divided by the sum (which is a
proxy for trading volume), then we can write that coefficient as
\begin{equation}
\sigma\frac{\left(  D-S\right)  ^{2}}{\left(  D+S\right)  ^{2}}.
\end{equation}
We can consider a more general case in which we write, for example, for
$\sigma=const,$%
\begin{equation}
d\log P\left(  t\right)  =\left(  \frac{D}{S}-1\right)  dt+\sigma\left(
\frac{D-S}{D+S}\right)  ^{p}dW
\end{equation}
where $p\in\mathbb{N}$ can be either even or odd. Note that we can write all
terms as functions of $f:=D/S-1,$ so $f+2=D/S+1>0$ since $D$ and $S$ are
positive, and we have%
\begin{equation}
d\log P\left(  t\right)  =fdt+\sigma\left(  \frac{f}{f+2}\right)  ^{p}dW.
\end{equation}
We write%
\begin{equation}
\mathbb{V}\left(  t\right)  :=\lim_{\Delta t\rightarrow0}\frac{V\left(
t,t+\Delta t\right)  }{\Delta t}=E\left[  \sigma\left(  t\right)
\left( \frac{f\left(  t\right)  }{f\left(  t\right)  +2}\right)^{p}\right]  ^{2}%
\end{equation}
If $f$ is deterministic and $\sigma$ is constant, we have upon
differentiation,%
\begin{equation}
\frac{d}{dt}\mathbb{V}\left(  t\right)  =4p\sigma^{2}\frac{f^{2p-1}}{\left[
f+2\right]  ^{2p+1}}\frac{df}{dt}%
\end{equation}
Recalling $f+2>0$ the sign of $\frac{d}{dt}\mathbb{V}$ depends only on
$f^{2p-1}df/dt.$ Notice that it makes no difference whether $p$ is even or odd.

If $f$ has a single maximum at $t_{m}\in\left(
t_{a},t_{b}\right)  $ such that $f\left(  t\right)  >0$ iff $t\in\left(
t_{a},t_{b}\right)  $, and $f<0$ iff $t\not \in \left[  t_{a},t_{b}\right]  $
then we have a relative maximum in $\mathbb{V}$ at $t_{m}$.

Hence, we see that if the coefficient of $dW$ is a deterministic term of the
form $\left(  \left(  D-S\right)  /(D+S)\right)  ^{p}$ and $f$ has a maximum, 
whether $p$ is even or odd (i.e., the coefficient
increases or decreases with excess demand), then the limiting volatility
$\mathbb{V}$ also has a maximum.
\end{example}

\begin{example}
Generalizing this concept further, we define a function $H\left(  z\right)  $
such that $H\left(  z\right)  >0\ $for all$~z\in\mathbb{R}$ \ and
\begin{equation}
\mathbb{\ }sgnH^{\prime}\left(  z\right)  =sgn\left(  z\right)  .
\end{equation}
We consider the stochastic equation, with $f$ deterministic%
\begin{equation}
d\log P=fdt+\sigma\left\{  H\left(  \frac{f}{f+2}\right)  \right\}  ^{1/2}dW
\end{equation}
so that $\mathbb{V}\left(  t\right)  =\sigma^{2}H\left(  \frac{f\left(
t\right)  }{f\left(  t\right)  +2}\right)  $ with $\sigma=const.$

While in principle, $f\left(  t\right)  :=D\left(  t\right)  /S\left(
t\right)  -1\in\left(  -1,\infty\right)  $, except under conditions that are
very far from equilibrium, one can assume $f\left(  t\right)  \in\left(
-a/2,a/2\right)  $ for some small $a,$ at least $a\in(0,1]$.

We compute
\begin{align}
\sigma^{-2}\frac{d}{dt}\mathbb{V}  &  =\frac{d}{dt}H\left(  \frac{f}%
{f+2}\right) \nonumber\\
&  =H^{\prime}\left(  \frac{f}{f+2}\right)  \frac{2}{\left(  f+2\right)  ^{2}%
}\frac{df}{dt}.
\end{align}
Based on this calculation, one concludes
if $f$ has a maximum, recalling that $f:=D/S-1$ is
positive near the maximum, then $d\mathbb{V}/dt$ has the same sign as $df/dt.$
So a maximum in $\mathbb{V}$ corresponds to a maximum in $f$, 
while $log\left(P\right)$ has its maximum at $t_{b}>t_{m}$.
\end{example}

\section{Supply and Demand as a function of valuation}

We consider the basic model (\ref{price}) now with the excess demand,
i.e., $D/S-1,$ depending on the valuation, $P_{a}\left(  t\right)  ,$ which
can be regarded either as a stochastic or deterministic function. It is now
commonly accepted in economics and finance that the trading price will often
stray from the fundamental valuation \cite{SSW, S}. We write the price
equation for the time evolution as
\begin{equation}
\frac{d}{dt}\log P\left(  t\right)  =\frac{D}{S}-1=\log\frac{P_{a}\left(
t\right)  }{P\left(  t\right)  }. \label{value}%
\end{equation}
The right hand side of equation (\ref{value}) is a linearization (as
discussed in Section 1.3) and the right hand side of has the
same linearization as $\left(  P_{a}-P\right)  /P$. The equation simply
expresses the idea that undervaluation is a motivation to buy, while
overvaluation is a motivation to sell, as one assumes in classical finance.
The non-classical feature is the absence of infinite arbitrage. Analogous to
Section 1.3, we write the stochastic version of (\ref{value}) as
\begin{equation}
d\log P\left(  t,\omega\right)  =\log\frac{P_{a}\left(  t,\omega\right)
}{P\left(  t,\omega\right)  }dt+\sigma\left(  t,\omega\right)  \left(
1+\log\frac{P_{a}\left(  t,\omega\right)  }{P\left(  t,\omega\right)
}\right)  dW\left(  t,\omega\right)  . \label{stochvalue}%
\end{equation}
At this point we allow both $P_{a}$ and $\sigma$ to be stochastic, with
$EP_{a}^{2}<C$ and $E\sigma^{2}<C$ but will specialize to given and
deterministic $P_{a}$ and $\sigma$ after the first result. We also assume
$1+\log\left(  P_{a}/P\right)  >0,$ i.e., $P_{a}/P>e^{-1},$ i.e., the
fundamental value, $P_{a},$ and trading price, $P,$ do not differ drastically.

\begin{notation}
Let $X:=\log P,$ $X_{a}:=\log P_{a},$ $y:=E\log P,\ y_{a}:=E\log P_{a},$
$z:=E\left(  \log P\right)  ^{2}.$ When $\log P_{a}$ and $\log P$ are
deterministic, we write lower case $x_{a}$ and $x,$ respectively.
\end{notation}

The equation $\left(  \ref{stochvalue}\right)  $ is short for the integral
form (using the notation above) for any $t_{2}>t_{1}>t_{0},$%
\begin{equation}
X\left(  t_{2}\right)  -X\left(  t_{1}\right)  =\int_{t_{1}}^{t_{2}}%
X_{a}-Xds+\int_{t_{1}}^{t_{2}}\sigma\left(  s,\omega\right)  \left(
1+X_{a}-X\right)  dW\left(  s\right)  . \label{logP}%
\end{equation}
Noting that $E\int f\left(  t\right)  dW=0,$ we find the expectation of
(\ref{logP}) as%
\begin{equation}
y\left(  t_{2}\right)  -y\left(  t_{1}\right)  =\int_{t_{1}}^{t_{2}}%
y_{a}\left(  s\right)  ds-\int_{t_{1}}^{t_{2}}y\left(  s\right)  ds
\label{expectationLogP}%
\end{equation}
i.e., one has the ODE, with $y_{0}:=y\left(  t_{0}\right)  :=E\left[  \log
P\left(  t_{0}\right)  \right]  ,$%
\begin{equation}
\frac{d}{dt}y\left(  t\right)  =y_{a}\left(  t\right)  -y\left(  t\right)
,\ \ \ y\left(  t_{0}\right)  :=y_{0} \label{Eqy}%
\end{equation}
This has the unique solution, for known $y_{a}\left(  t\right)  ,$
\begin{equation}
y\left(  t\right)  =e^{t_{0}-t}y\left(  t_{0}\right)  +e^{-t}\int_{t_{0}}%
^{t}y_{a}\left(  s\right)  e^{s}ds. \label{Solny}%
\end{equation}

Note that if we eliminate randomness altogether, i.e., $\sigma:=0$ and
deterministic $P_{a}\left(  t\right)  $,
\begin{equation}
\frac{d}{dt}\log P\left(  t\right)  =\log\frac{P_{a}\left(  t\right)
}{P\left(  t\right)  },
\end{equation}
with solution%
\begin{equation}
x\left(  t\right)  =e^{t_{0}-t}x\left(  t_{0}\right)  +e^{-t}\int_{t_{0}}%
^{t}e^{s}x_{a}\left(  s\right)  ds.
\end{equation}
where $x\left(  t\right)  :=\log P\left(  t\right)  $ and $x_{a}\left(
t\right)  :=\log P_{a}\left(  t\right)  $. We note that the solution of
$y\left(  t\right)  =E\log P\left(  t\right)  $ of $\log P$ in terms of
$y_{a}\left(  t\right)  =E\log P_{a}\left(  t\right)  $ is the same as $\log
P\left(  t\right)  $ in terms of $\log P_{a}\left(  t\right)  $, i.e. both
expected value and deterministic $\log P$ satisfy the same equation.

\subsection{The stochastic problem}

We write the SDE $\left(  \ref{stochvalue}\right)  $ as%
\begin{equation}
dX=\left(  X_{a}-X\right)  dt+\sigma\left(  1+X_{a}-X\right)  dW. \label{sde}%
\end{equation}
We say $X$ is a solution to a SDE if $X\in H_{2}\left[  0,T\right]  $ and
solves the integral version of the SDE for almost every $\omega\in\Omega$. The
solution to the stochastic equation (\ref{stochvalue}), $X\left(
t,\omega\right)  $ is unique, belongs to $H_{2}\left[  0,T\right]  $ and is
continuous in $t\in\left[  0,T\right]  $ for almost every $\omega\in\Omega$
(\cite{SH} p. 94). We denote the remaining set $\Gamma,$ so that $X\left(
t,\omega\right)  $ is continuous in $t$ for all $\omega\in\Omega
\ \backslash\ \Gamma.$ One has by basic measure theory (e.g., \cite{R}), that
for any measurable function such as $X$ or $X^{2}$ one has%
\begin{align}
E\int_{t}^{t+\Delta t}X\left(  s,\omega\right)  ds  &  =\int_{\Omega}\int
_{t}^{t+\Delta t}X\left(  s,\omega\right)  dsdP\left(  \omega\right) \nonumber\\
&  =\int_{\Omega\ \backslash\ \Gamma}\int_{t}^{t+\Delta t}X\left(  s,\omega\right)  dsdP\left(
\omega\right)  +\int_{\Gamma}\int_{t}^{t+\Delta t}X\left(  s,\omega\right)  dsdP\left(  \omega\right)
.
\end{align}
Thus from here on we can ignore the set $\Gamma$ and assume that $X\left(
t,\omega\right)  $ is continuous when the expectation value is computed .

Next, using (\ref{expectationLogP}) we compute the variance, of $\Delta
X:=X\left(  t+\Delta t,\omega\right)  -X\left(  t,\omega\right)  $ and later
we will determine the terms that vanish upon dividing by $\Delta t,$%
\begin{align}
V\left(  t,t+\Delta t\right)   &  :=E\left[  X\left(  t+\Delta t\right)
-EX\left(  t\right)  \right]  ^{2}-\left(  E\left[  X\left(  t+\Delta
t\right)  -X\left(  t\right)  \right]  \right)  ^{2}\nonumber\\
&  =E\left[  \int_{t}^{t+\Delta t}X_{a}-Xds+\int_{t}^{t+\Delta t}\sigma\left(
1+X_{a}-X\right)  dW\left(  s\right)  \right]  ^{2}\label{var}\nonumber\\
&  -\left(  E\left[  \int_{t}^{t+\Delta t}X_{a}-Xds+\int_{t}^{t+\Delta
t}\sigma\left(  1+X_{a}-X\right)  dW\left(  s\right)  \right]  \right)
^{2}\ .
\end{align}
Note that with $\Delta X:=X\left(  t+\Delta t\right)  -X\left(  t\right)  $ we
have
\begin{align}
V\left(  t,t+\Delta t\right)   &  =Var\left[  X\left(  t+\Delta t\right)
-X\left(  t\right)  \right]  =Var\left[  \log\frac{P\left(  t+\Delta t\right)
}{P\left(  t\right)  }\right] \nonumber\\
&  =Var\left[  \log\left(  \frac{\Delta P}{P}+1\right)  \right]  \simeq
Var\left[  \frac{\Delta P}{P}\right]  \ . \label{varapp}%
\end{align}
so that $V\left(  t,t+\Delta t\right)  $ is essentially a measure of the
variance of relative price change. Since $E\int_{t}^{t+\Delta t}\sigma\left(
1+X_{a}-X\right)  dW\left(  s\right)  =0$ one has%
\begin{align}
V\left(  t,t+\Delta t\right)   &  =E\left[  \int_{t}^{t+\Delta t}%
X_{a}-Xds+\int_{t}^{t+\Delta t}\sigma\left(  1+X_{a}-X\right)  dW\left(
s\right)  \right]  ^{2}\nonumber\\
&  -\left(  E\int_{t}^{t+\Delta t}X_{a}-Xds\right)  ^{2}\nonumber\\
&  =V_{1}\left(  t,t+\Delta t\right)  +V_{2}\left(  t,t+\Delta t\right)
+V_{3}\left(  t,t+\Delta t\right)
\end{align}
where%
\begin{align}
V_{1}\left(  t,t+\Delta t\right)   &  :=E\left(  \int_{t}^{t+\Delta t}%
X_{a}-Xds\right)  ^{2}-\left(  E\int_{t}^{t+\Delta t}X_{a}-Xds\right)
^{2},\nonumber\\
V_{2}\left(  t,t+\Delta t\right)   &  :=2E\left[  \int_{t}^{t+\Delta t}%
X_{a}-Xds\int_{t}^{t+\Delta t}\sigma\left(  1+X_{a}-X\right)  dW\left(
s\right)  \right] \nonumber\\
V_{3}\left(  t,t+\Delta t\right)   &  :=E\left(  \int_{t}^{t+\Delta t}%
\sigma\left(  1+X_{a}-X\right)  dW\left(  s\right)  \right)  ^{2}\nonumber\\
&  =\int_{t}^{t+\Delta t}E\left[  \sigma\left(  1+X_{a}-X\right)  \right]
^{2}ds \label{v3}%
\end{align}

\begin{lemma}
\label{Vest}Let $X$ be a solution to the SDE $\left(  \ref{sde}\right)  $ with
$\sigma\left(  t,\omega\right)  $ and $X_{a}\left(  t,\omega\right)  $
continuous for all $t\in\left[  0,T\right]  $ and all $\omega\in\Omega,\ $with
bounded second moments. Then
\end{lemma}

$\left(  i\right)  $ $\left\vert V_{1}\left(  t,t+\Delta t\right)  \right\vert
\leq C\left(  \Delta t\right)  ^{2}$ so $\lim_{\Delta t\rightarrow0}%
V_{1}\left(  t,t+\Delta t\right)  /\Delta t=0,$ and,

$\left(  ii\right)  $ $\left\vert V_{2}\left(  t,t+\Delta t\right)
\right\vert \leq C\left(  \Delta t\right)  ^{3/2}$ so $\lim_{\Delta
t\rightarrow0}V_{1}\left(  t,t+\Delta t\right)  /\Delta t=0.$

\begin{proof}
$\left(  i\right)  \left(  a\right)  $ We consider the first term in $V_{1},$
namely,%
\begin{align}
E\left(  \int_{t}^{t+\Delta t}X_{a}-Xds\right)  ^{2}  &  =\int_{\Omega}\left(
\int_{t}^{t+\Delta t}X_{a}-Xds\right)  ^{2}dP\left(  \omega\right) \nonumber\\
&  =\int_{\Omega\ \backslash\ \Gamma}\left(  \int_{t}^{t+\Delta t}%
X_{a}-Xds\right)  ^{2}dP\left(  \omega\right)
\end{align}
where we have omitted the set of measure zero, $\Gamma,$ outside of which $X$
is continuous in $t$ on a closed bounded interval. Hence, one can bound the
integrand by $C\left(  \Delta t\right)  ^{2}.$ Thus we have%
\begin{equation}
E\left(  \int_{t}^{t+\Delta t}X_{a}-Xds\right)  ^{2}\leq C\left(  \Delta
t\right)  ^{2}.
\end{equation}

$\left(  i\right)  \left(  b\right)  $ Similarly the second term can be
bounded as%
\begin{align}
\left(  E\int_{t}^{t+\Delta t}X_{a}-Xds\right)  ^{2}  &  =\left(  \int
_{\Omega\ \backslash\ \Gamma}\left(  \int_{t}^{t+\Delta t}X_{a}-Xds\right)
dP\left(  \omega\right)  \right)  ^{2}\nonumber\\
&  \leq C\left(  \Delta t\right)  ^{2}.
\end{align}
Hence, part $\left(  i\right)  $ of the lemma has been proven.

$\left(  ii\right)  $ Using the Schwarz inequality on the second term we have%
\begin{align}
\frac{1}{2}V_{2}\left(  t,t+\Delta t\right)   &  =E\left\{  \left(  \int
_{t}^{t+\Delta t}X_{a}-Xds\right)  \left(  \int_{t}^{t+\Delta t}\sigma\left(
1+X_{a}-X\right)  dW\left(  s\right)  \right)  \right\} \nonumber\\
&  \leq\left\{  E\left(  \int_{t}^{t+\Delta t}X_{a}-Xds\right)  ^{2}\right\}
^{1/2}\left\{  E\left(  \int_{t}^{t+\Delta t}\sigma\left(  1+X_{a}-X\right)
dW\left(  s\right)  \right)  ^{2}\right\}  ^{1/2}.
\end{align}
Using continuity properties, we have the following bound on the first term,
\begin{equation}
\left\{  E\left(  \int_{t}^{t+\Delta t}X_{a}-Xds\right)  ^{2}\right\}
^{1/2}\leq C\left(  \Delta t\right)  .
\end{equation}
For the second we use the basic property used above,%
\begin{align}
\left\{  E\left(  \int_{t}^{t+\Delta t}\sigma\left(  1+X_{a}-X\right)
dW\left(  s\right)  \right)  ^{2}\right\}  ^{1/2}  &  =\left\{  \int
_{t}^{t+\Delta t}E\left[  \sigma\left(  1+X_{a}-X\right)  \right]
^{2}ds\right\}  ^{1/2}\nonumber\\
&  =\left\{  \int_{\Omega\ \backslash\ \Gamma}\int_{t}^{t+\Delta t}E\left[
\sigma\left(  1+X_{a}-X\right)  \right]  ^{2}ds\right\}  ^{1/2}\nonumber\\
&  \leq C\left(  \Delta t\right)  ^{1/2}.
\end{align}

Hence, the proof of the second part of the lemma follows from the following
bound:%
\begin{align}
V_{2}\left(  t,t+\Delta t\right)   &  \leq\left\{  E\left(  \int_{t}^{t+\Delta
t}X_{a}-Xds\right)  ^{2}\right\}  ^{1/2}\left\{  E\left(  \int_{t}^{t+\Delta
t}\sigma\left(  1+X_{a}-X\right)  dW\left(  s\right)  \right)  ^{2}\right\}
^{1/2}\nonumber\\
&  \leq C\left(  \Delta t\right)  ^{3/2}%
\end{align}
This proves the second part of the Lemma.
\end{proof}

\bigskip

Thus, Lemma \ref{Vest} indicates that in analyzing $V\left(  t,t+\Delta
t\right)  /\Delta t$ in the limit of $\Delta t\rightarrow0$ amounts to
analyzing $V_{3}\left(  t,t+\Delta t\right)  /\Delta t.$

\bigskip

At this point we assume that both $P_{a}$ and $\sigma$ are deterministic but
need not be constant in time, and we now use lower case, $x_{a}:=$ $\log
P_{a}$ .

\begin{lemma}
Let $\sigma$ and $P_{a}$ be deterministic, and $X\left(  t\right)  $ as
solution to the SDE $\left(  \ref{sde}\right)  $. Then%
\begin{equation}
V_{3}\left(  t,t+\Delta t\right)  =\int_{t}^{t+\Delta t}\sigma^{2}\left[
1+x_{a}\left(  s\right)  -EX\left(  s\right)  \right]  ^{2}ds+\int
_{t}^{t+\Delta t}\sigma^{2}Var\left[  X\left(  s\right)  \right]  ds.
\end{equation}

\end{lemma}

\bigskip

\begin{proof}
Using the expression $\left(  \ref{v3}\right)  $ above, the identity follows
upon adding and subtracting $EX^{2}\left(  s\right)  $ in the integrand.
\end{proof}

\bigskip

\begin{lemma}
Let $\sigma$ and $P_{a}$ be deterministic and continuous. Then%
\[
\mathbb{V}\left(  t\right)  :=\lim_{\Delta t\rightarrow0}\frac{V\left(
t,t+\Delta t\right)  }{\Delta t}=\lim_{\Delta t\rightarrow0}\frac{V_{3}\left(
t,t+\Delta t\right)  }{\Delta t}%
\]%
\begin{align}
&  =\lim_{\Delta t\rightarrow0}\frac{1}{\Delta t}\left\{  \int_{t}^{t+\Delta
t}\sigma^{2}\left[  1+x_{a}-y\right]  ^{2}+Var\left[  X\right]  ds\right\}
\nonumber\\
&  =\sigma^{2}\left[  1+x_{a}-y\right]  ^{2}+Var\left[  X\right]  .
\label{boldV}%
\end{align}

\end{lemma}

\bigskip

Next, we will compute $Var\left[  X\right]  $ starting with $E\left[
X^{2}\right]  $ and assuming that $P_{a}$ and $\sigma$ are deterministic.

\begin{lemma}
\label{LemdZ}Let $\sigma$ and $x_{a}$ be deterministic and continuous. Then
$z\left(  t\right)  :=EX^{2}\left(  t\right)  $ satisfies the ODE%
\begin{align}
\frac{dz}{dt}  &  =\left(  \sigma^{2}-2\right)  z+\left(  2-2\sigma
^{2}\right)  x_{a}y-2\sigma^{2}y+\sigma^{2}\left(  1+x_{a}\right)
^{2}\nonumber\\
z\left(  t_{0}\right)   &  =y\left(  t_{0}\right)  ^{2}=:y_{0}^{2}\ .
\label{ODEz}%
\end{align}

\end{lemma}

\begin{proof}
The stochastic process for $X\left(  t\right)  $, i.e., $\left(
\ref{sde}\right)  $ can be written
\[
A\left(  t,\omega\right)  :=\left(  x_{a}-X\right)  ,\ \ B\left(
t,\omega\right)  :=\sigma\left(  1+x_{a}-X\right)
\]%
\begin{equation}
dX\left(  t,\omega\right)  =A\left(  t,\omega\right)  dt+B\left(
t,\omega\right)  dW\left(  t\right)
\end{equation}
Ito's formula provides the differential for a smooth function of $X$ as%
\begin{align}
df\left(  X\left(  t\right)  ,t\right)   &  =\left[  \frac{\partial f(X\left(
t\right)  ,t)}{\partial t}+A\left(  t\right)  \frac{\partial f(X\left(
t\right)  ,t)}{\partial x}+\frac{B^{2}\left(  t\right)  }{2}\frac{\partial
^{2}f\left(  X\left(  t\right)  ,t\right)  }{\partial x^{2}}\right]
dt\nonumber\\
&  +B\left(  t\right)  \frac{\partial f\left(  X\left(  t\right)  ,t\right)
}{\partial x}dW\left(  t\right)  .
\end{align}
For $f\left(  x\right)  :=x^{2}$ we have then from Ito's formula,%
\begin{align}
dX^{2}  &  =\left[  \left(  \sigma^{2}-2\right)  X^{2}+\left(  2-2\sigma
^{2}\right)  x_{a}X-2\sigma^{2}X+\sigma^{2}\left(  1+x_{a}\right)
^{2}\right]  dt\nonumber\\
&  +\sigma\left(  1+x_{a}-X\right)  \left(  2X\right)  dW
\end{align}
Hence, we can write in the usual way, as $EdW$ vanishes:%
\begin{equation}
E\left[  X^{2}\left(  t\right)  -X^{2}\left(  t_{0}\right)  \right]
=\int_{t_{0}}^{t}\left(  \sigma^{2}-2\right)  EX^{2}+\left(  2-2\sigma
^{2}\right)  x_{a}EX-2\sigma^{2}EX+\sigma^{2}\left(  1+x_{a}\right)  ^{2}ds
\end{equation}
Using the notation $y\left(  t\right)  :=E\left(  \log P\right)  $ and
$z\left(  t\right)  :=E\left(  \log P\right)  ^{2}$ we have
\begin{equation}
z\left(  t\right)  -z\left(  t_{0}\right)  =\int_{t_{0}}^{t}\left(  \sigma
^{2}-2\right)  z\left(  t\right)  +\left(  2-2\sigma^{2}\right)  x_{a}y\left(
t\right)  -2\sigma^{2}y\left(  t\right)  +\sigma^{2}\left(  1+x_{a}\right)
^{2}ds.
\end{equation}
Differentiation with respect to $t$ yields the result and proves the lemma.

\end{proof}

In the sequel, we assume for simplicity that $\sigma$ is constant in time, and
$x_{a}\left(  t\right)  $ is deterministic and smooth. We can solve for $z$
directly but it will be more illuminating if we write the solution in the
following form.

\begin{lemma}
\textbf{\label{LemmaZ0}}Let $x_{a}$ be a continuous function. The unique
solution to
\begin{align}
\frac{dz_{0}}{dt}  &  =-2z_{0}+2x_{a}y\label{z0}\\
z_{0}\left(  t_{0}\right)   &  :=y\left(  t_{0}\right)  ^{2} \label{z0IC}%
\end{align}
is given by $z_{0}\left(  t\right)  =y\left(  t\right)  ^{2}.$
\end{lemma}

\begin{proof}
Note that $x_{a}=y_{a}=EX_{a}$ since $X_{a}$ is deterministic under our
current assumption. We know that $y\left(  t\right)  $ is a solution to the
equation%
\begin{equation}
\frac{d}{dt}y\left(  t\right)  =y_{a}\left(  t\right)  -y\left(  t\right)
,\ \ \ y\left(  t_{0}\right)  :=y_{0}%
\end{equation}
so we can substitute $x_{a}=y^{\prime}+y$ into $\left(  22\right)  $ and
obtain%
\begin{equation}
z_{0}^{\prime}+2z_{0}=2yy^{\prime}+2y^{2}=2y\left(  y^{\prime}+y\right)
=2x_{a}y.
\end{equation}
Hence, $z_{0}\left(  t\right)  :=y\left(  t\right)  ^{2}$ solves $\left(
\ref{z0}\right)  ,\left(  \ref{z0IC}\right)  $ and from basic ODE theory, the
solution is unique so long as $x_{a}$ is continuous.
\end{proof}

\begin{lemma}
\label{LemmaZ1}The unique solution to $\left(  \ref{ODEz}\right)  $ is given
by
\begin{equation}
z\left(  t\right)  :=z_{0}\left(  t\right)  +\sigma^{2}z_{1}\left(  t\right)
=y\left(  t\right)  ^{2}+\sigma^{2}z_{1}\left(  t\right) \label{z0}
\end{equation}
with $z_{1}\left(  t\right)  $ defined by%
\begin{equation}
z_{1}\left(  t\right)  =\int_{t_{0}}^{t}e^{\left(  2-\sigma^{2}\right)
\left(  s-t\right)  }\left[  y\left(  s\right)  -\left(  1+x_{a}\left(
s\right)  \right)  \right]  ^{2}ds. \label{z1}%
\end{equation}

\end{lemma}

\begin{proof}
Let $z_{1}$ be defined by $z\left(  t\right)  =z_{0}\left(  t\right)
+\sigma^{2}z_{1}\left(  t\right)  =y\left(  t\right)  ^{2}+\sigma^{2}%
z_{1}\left(  t\right)  .$ Substituting into $\left(  \ref{ODEz}\right)  $
yields%
\begin{align}
z_{0}^{\prime}+\sigma^{2}z_{1}^{\prime}  &  =\left(  \sigma^{2}-2\right)
\left(  z_{0}+\sigma^{2}z_{1}\right)  +\left(  2-2\sigma^{2}\right)
x_{a}y-2\sigma^{2}y+\sigma^{2}\left(  1+x_{a}\right)  ^{2}\nonumber\\
&  =\sigma^{2}z_{0}-2z_{0}+\left(  \sigma^{2}-2\right)  \sigma^{2}%
z_{1}+\left(  2-2\sigma^{2}\right)  x_{a}y-2\sigma^{2}y+\sigma^{2}\left(
1+x_{a}\right)  ^{2}%
\end{align}
so that the terms $z_{0}^{\prime}$ and $-2z_{0}+2x_{a}y$ vanish from both
sides.. Using $z_{0}=y^{2}$ we have left, upon dividing by $\sigma^{2},$ the
equation for $z_{1}$%
\begin{equation}
z_{1}^{\prime}+\left(  2-\sigma^{2}\right)  z_{1}=\left[  y-\left(
1+x_{a}\right)  \right]  ^{2},%
\end{equation}
and elementary methods yield the solution (\ref{z0} - \ref{z1}).
\end{proof}

Note that although $\sigma\in\mathbb{R}$ \ was used in this proof, comparable
result can be obtained in the general case in which $\sigma$ is a 
continuous and deterministic function.

\bigskip

Thus, Lemmas \ref{LemmaZ0} and \ref{LemmaZ1} yield the following identity for
$Var\left[  X\left(  t\right)  \right]  .$

\begin{theorem}
\label{ThmVar}Let $\sigma\in\mathbb{R}$ and $x_{a}\left(  t\right)  $ be
deterministic and continuous. Let $c:=\left(  2-\sigma^{2}\right)  $ and%
\begin{equation}
\sigma^{2}I\left(  t,t+\Delta t\right)  :=Var\left[  X\left(  t+\Delta
t\right)  \right]  -Var\left[  X\left(  t\right)  \right]  .
\end{equation}
\begin{equation}
w\left(  s\right)  :=\left[  1+x_{a}\left(  s\right)  -y\left(  s\right)
\right]  ^{2}.
\end{equation}
Then one has the following identities:
\begin{align}
Var\left[  X\left(  t\right)  \right]   &  =\sigma^{2}\int_{t_{0}}%
^{t}e^{c\left(  s-t\right)  }\left[  y\left(  s\right)  -\left(
1+x_{a}\left(  s\right)  \right)  \right]  ^{2}ds\\
I\left(  t,t+\Delta t\right)   &  =\int_{t}^{t+\Delta t}e^{c\left(
s-t\right)  }w\left(  s\right)  ds.
\end{align}

\end{theorem}

\bigskip

\begin{proof}
The identities follow immediately from Lemma \ref{LemmaZ1} and the definition
of variance in terms of $z$ and $y.$ I.e.,
\begin{align}
Var\left[  X\left(  t\right)  \right]   &  =E\left[  X\left(  t\right)
\right]  ^{2}-\left[  EX\left(  t\right)  \right]  ^{2}\nonumber\\
&  =z\left(  t\right)  -y\left(  t\right)  ^{2}=\sigma^{2}z_{1}\left(
t\right) \nonumber\\
&  =\sigma^{2}\int_{t_{0}}^{t}e^{\left(  2-\sigma^{2}\right)  \left(
s-t\right)  }\left[  1+x_{a}\left(  s\right)  -y\left(  s\right)  \right]
^{2}ds.
\end{align}

\end{proof}

\bigskip

\begin{remark}
The maximum value of $Var\left[  X\left(  t+\Delta t\right)  \right]
-Var\left[  X\left(  t\right)  \right]  $ occurs for $t$ such that the average
weighted value of $w\left(  s\right)  $ with exponential weighting of $\left(
2-\sigma^{2}\right)  $ is maximal on $\left(  t,t+\Delta t\right)  .$
\end{remark}

Using the lemmas above, we obtain directly the following result.

\bigskip

\begin{theorem}
\label{ThmQ} Let $x_{a}$ be continuous. Then we have the identities,
\begin{align}
\lim_{\Delta t\rightarrow0}\sigma^{-2}\left(  \Delta t\right)  ^{-1}V\left(
t,t+\Delta t\right)   &  =\lim_{\Delta t\rightarrow0}\sigma^{-2}\left(  \Delta
t\right)  ^{-1}V_{3}\left(  t,t+\Delta t\right) \nonumber\\
&  =w\left(  t\right)  +Var\left[  X\left(  t\right)  \right]
\ \ \ \ \ \label{w}\nonumber\\
i.e.,\ \ \ \ \ \ \ \ \ \ \ \ \ \ \ \ \ \sigma^{-2}\mathbb{V}\left(  t\right)
&  =w\left(  t\right)  +\sigma^{2}\int_{t_{0}}^{t}e^{\left(  2-\sigma\right)
^{2}\left(  s-t\right)  }w\left(  s\right)  ds
\end{align}%
\begin{align}
Q\left(  t\right)   &  :=\frac{d}{dt}\lim_{\Delta t\rightarrow0}\sigma
^{-2}\frac{V\left(  t,t+\Delta t\right)  }{\Delta t}=\sigma^{-2}\frac{d}%
{dt}\mathbb{V}\left(  t\right) \label{Q}\nonumber\\
&  =w^{\prime}\left(  t\right)  +\sigma^{2}w\left(  t\right)  -\sigma
^{2}\left(  2-\sigma^{2}\right)  \int_{t_{0}}^{t}e^{\left(  2-\sigma\right)
^{2}\left(  s-t\right)  }w\left(  s\right)  ds.
\end{align}

\end{theorem}

\bigskip

\section{Market extrema}

The main objective of this section is to apply the results above understand
the temporal relationship between the extrema of the (log) fundamental value,
$x_{a}\left(  t\right)  $, and the expected (log) trading price, $y\left(
t\right)  .$

\subsection{Price Maxima}

\begin{notation}
Let $t_{0}$ be the initial time, and $t_{m}$ be defined by $x_{a}^{\prime
}\left(  t_{m}\right)  =0,$ i.e., the time at which the fundamental value,
$x_{a},$ attains its maximum. The time $t_{\ast}$ is defined as the first time
at which $y^{\prime}\left(  t_{\ast}\right)  =x_{a}\left(  t_{\ast}\right)
-y\left(  t_{\ast}\right)  $ vanishes, and the curves $x_{a}\left(  t\right)
$ and $y\left(  t\right)  $ first intersect.
\end{notation}

\begin{notation}
Let $\hat{x}_{a}\left(  t\right)  :=e^{t}x_{a}\left(  t\right)  $, $\hat
{y}\left(  t\right)  :=e^{t}y\left(  t\right)  ,$ $\hat{y}_{0}:=e^{t_{0}}%
\hat{y}\left(  t_{0}\right)  .$
\end{notation}

\textbf{Condition} $\sigma$ . Let $\sigma\in\left(  0,1\right)  $ be a
constant, so $c:=2-\sigma^{2}\in\left(  1,2\right)  .$

We will assume this condition throughout, though some results are valid
without it.

\textbf{Condition} $C.$ $\left(  i\right)  $ The function $x_{a}:[t_{0}%
,\infty)\rightarrow\left(  0,\infty\right)  $ has the property that for some
$t_{m}\in\left(  0,\infty\right)  $ one has

$\left(  i\right)  \ \ \ \ \ \ x_{a}^{\prime}\left(  t\right)
>0\ \ if\ \ t<t_{m};\ \ x_{a}^{\prime}\left(  t_{m}\right)  =0;\ \ x_{a}%
^{\prime}\left(  t\right)  <0\ \ if\ t>t_{m}.$

$\left(  ii\right)  $ \ Let $y\left(  t_{0}\right)  =:$ $y_{0}\in\left(
0,\infty\right)  $ one has%
\begin{equation}
x_{a}\left(  t_{0}\right)  -x_{a}^{\prime}\left(  t_{0}\right)  <y_{0}%
<x_{a}\left(  t_{0}\right)  .
\end{equation}

$\left(  iii\right)  $ For some $\delta,m_{1}\in\left(  0,\infty\right)  $ one
has
\begin{equation}
-x_{a}^{\prime}\left(  t\right)  >m_{1}>0\ \ if\ \ t>t_{m}+\delta.
\end{equation}

\textbf{Remarks. }We set $y_{0}=:y\left(  t_{0}\right)  ,$ so the two
inequalities in\textbf{ }Condition $C\left(  ii\right)  $ state that initially
(i.e., at $t_{0}$) the price is below the fundamental value, i.e.,
undervaluation ($y\left(  t_{0}\right)  =y_{0}<x_{a}\left(  t_{0}\right)$).
Using the ODE $y^{\prime}=x_{a}-y$ one has that the first inequality in
Condition $C\left(  ii\right)  $ is equivalent to \ $x_{a}^{\prime}\left(
t_{0}\right)  >y^{\prime}\left(  t_{0}\right)  >0$ stipulating that the
valuation has begun to increase relative to trading price. Condition $C\left(
iii\right)  $ can be relaxed to some extent although the condition then
appears more technical.

\bigskip

\textbf{Condition} $E.$ With $t_{\ast}$ be defined as above, assume
$2x_{a}^{\prime}(t_{\ast})+\sigma^{2}e^{c\left(  t_{0}-t_{\ast}\right)  }<0.$

\bigskip

\textbf{Remarks.} Note that this condition is satisfied automatically if
$t_{0}\rightarrow-\infty.$ So long as there is an interval $\left(
t_{m},t_{\ast}\right)  $ on which $x_{a}^{\prime}\left(  t_{\ast}\right)  <$
$-\sigma^{2}e^{c\left(  t_{0}-t_{\ast}\right)  }$ (the latter is exponentially
small if $t_{\ast}-t_{0}>>1$) the Condition $E$ will be satisfied.

Recalling that $y\left(  t\right)  $ is given by $\left(  \ref{Solny}\right)
$, i.e.,%
\begin{equation}
\hat{y}\left(  t\right)  =\hat{y}\left(  t_{0}\right)  +\int_{t_{0}}^{t}%
\hat{x}_{a}\left(  s\right)  ds. \label{Eqnyhat}%
\end{equation}
since $y_{a}=x_{a}$ as the latter is deterministic.

Initially, we have from $C\left(  ii\right)  $ that $x_{a}\left(
t_{0}\right)  >y\left(  t_{0}\right)  .$ We want to first prove that $y$
intersects with $x_{a}$ at some value $t_{\ast}$ and that this value $t_{\ast
}$ occurs after $t_{m}$ (i.e., the time at which $x_{a}$ has its peak).

\begin{theorem}
\label{Thmtopt*}Assume that $C$ holds. Then there exists a least value
$t_{\ast}\in\left(  t_{m},\infty\right)  $ such that for $t<t_{\ast}$ one has
$y\left(  t\right)  <x_{a}\left(  t\right)  ,$ and, $y\left(  t_{\ast}\right)
<x_{a}\left(  t_{\ast}\right)  ,$ i.e.,
\begin{equation}
\hat{y}\left(  t_{\ast}\right)  =\hat{y}_{0}+\int_{t_{0}}^{t_{\ast}}\hat
{x}_{a}\left(  s\right)  ds=\hat{x}_{a}\left(  t_{\ast}\right)  .
\label{Eqnyhat*}%
\end{equation}
Since $y^{\prime}=x_{a}-y,$ the maximum of $y$ is attained at $t_{\ast}$.
\end{theorem}

\begin{proof}
Let $I\left(  t\right)  :=\hat{x}_{a}\left(  t\right)  -\hat{y}_{0}%
-\int_{t_{0}}^{t}\hat{x}_{a}\left(  s\right)  ds,$ so $I\left(  t_{0}\right)
>0$ by condition $C\left(  ii\right)  .$ Computing the derivative and using
Condition $C\left(  i\right)  $ yields%
\begin{equation}
I^{\prime}\left(  t\right)  =\hat{x}_{a}^{\prime}\left(  t\right)  -\hat
{x}_{a}\left(  t\right)  =e^{t}x_{a}^{\prime}\left(  t\right)  >0\ if\ t<t_{m}%
\ .
\end{equation}
Hence, one has $I\left(  t\right)  <0$ if $t<t_{m}\ .$ On the other hand, by
Condition $C\left(  iii\right)  $, when $t>t_{m}+\delta$ one has%
\begin{equation}
I^{\prime}\left(  t\right)  =e^{t}x_{a}^{\prime}\left(  t\right)  \leq
e^{t_{m}}\left(  -m_{1}\right)
\end{equation}
so that $I\left(  t_{\ast}\right)  =0$ for some finite $t_{\ast}>t_{m}.$
\end{proof}
\begin{lemma}
Under $C\left(  i\right)  ,\left(  ii\right)  $ there exists a $t_{1}%
\in\left(  t_{0},t_{m}\right)  $ such that $w^{\prime}\left(  t_{1}\right)
=0$, $w^{\prime}\left(  t\right)  >0$ if $t\in\lbrack t_{0},t_{1}),$ and
$w^{\prime}\left(  t\right)  <0$ if $t_{m}<t<t_{\ast}\ .$ Consequently, we have
\begin{equation}
t_{0}<t_{1}<t_{m}<t_{\ast}\ . \label{TopOrdert}%
\end{equation}
\end{lemma}

\begin{proof}
Recall $\left(  \ref{w}\right)  $ and note $w^{\prime}=2\left[  1+x_{a}%
-y\right]  \left(  x_{a}^{\prime}-y^{\prime}\right)  ,$ whose sign is
determined by
\begin{equation}
S\left(  t\right)  :=x_{a}^{\prime}\left(  t\right)  -y^{\prime}\left(
t\right)  =x_{a}^{\prime}\left(  t\right)  -x_{a}\left(  t\right)  +y\left(
t\right)
\end{equation}
when $t<t_{\ast}$ [i.e., when $x_{a}\left(  t\right)  >y\left(  t\right)  $].
For $t_{0}$ we have from $C\left(  ii\right)  $ that $S\left(  t_{0}\right)
>0.$

For \ $t_{m}<t\,<t_{\ast}$ we have from $C\left(  i\right)  $ that
$x_{a}^{\prime}\left(  t\right)  <0$ while $y^{\prime}\left(  t\right)
=x_{a}\left(  t\right)  -y\left(  t\right)  >0$ as noted earlier in the proof,
yielding%
\begin{equation}
S\left(  t\right)  =x_{a}^{\prime}\left(  t\right)  -x_{a}\left(  t\right)
+y\left(  t\right)  <0.
\end{equation}
By continuity, there exists a $t_{1}\in\left(  t_{0},t_{m}\right)  $ such that
$S\left(  t_{1}\right)  =0$ and $S\left(  t\right)  >0$ for $t<t_{1}.$ I.e.,
$t_{1}$ is the first crossing for $S\left(  t\right)  $ and hence for
$w\left(  t\right)  .$ The ordering $\left(  \ref{TopOrdert}\right)  $ thus follows.
\end{proof}

\begin{lemma}
Assuming Condition $C,$ one has $Q\left(  t_{1}\right)  >0.$
\end{lemma}

\begin{proof}
Since $t_{1}<t_{m}<t_{\ast}$ one has $x_{a}\left(  t_{1}\right)  >y\left(
t_{1}\right)  $ and consequently $w\left(  t_{1}\right)  $ exceeds $1$ and is
thus positive. Hence, we can replace $w\left(  s\right)  $ by $w\left(
t_{1}\right)  $ in the integral, and factor, in order to obtain the inequality%
\begin{align}
Q\left(  t_{1}\right)   &  \geq0+\sigma^{2}w\left(  t_{1}\right)  -\sigma
^{2}c\int_{t_{0}}^{t_{1}}e^{c\left(  s-t_{1}\right)  }w\left(  t_{1}\right)
ds\nonumber\\
&  =\sigma^{2}w\left(  t_{1}\right)  \left\{  1-\left(  1-e^{c\left(
t_{0}-t_{1}\right)  }\right)  \right\} \nonumber\\
&  =\sigma^{2}w\left(  t_{1}\right)  e^{c\left(  t_{0}-t_{1}\right)  }>0.\
\end{align}

\end{proof}

\begin{lemma}
If Conditions $C$ and $E$ hold, then $Q\left(  t_{\ast}\right)  <0.$
\end{lemma}

\begin{proof}
We write%
\begin{equation}
Q\left(  t_{\ast}\right)  =w^{\prime}\left(  t_{\ast}\right)  +\sigma
^{2}w\left(  t_{\ast}\right)  -\sigma^{2}c\int_{t_{0}}^{t_{\ast}}e^{c\left(
s-t_{\ast}\right)  }w\left(  s\right)  ds,
\end{equation}
and note that for any $t\leq t_{\ast}$ one has $x_{a}\left(  t\right)
>y\left(  t\right)  $ by Thm \ref{Thmtopt*}. Consequently, we have the
inequality%
\begin{equation}
w\left(  t\right)  =\left[  1+x_{a}\left(  t\right)  -y\left(  t\right)
\right]  ^{2}\geq1=w\left(  t_{\ast}\right)  .
\end{equation}
By using this mimimum value of $w$ that is subtracted, we have
\begin{equation}
Q\left(  t_{\ast}\right)  \leq w^{\prime}\left(  t_{\ast}\right)  +\sigma
^{2}w\left(  t_{\ast}\right)  -\sigma^{2}c\int_{t_{0}}^{t_{\ast}}e^{c\left(
s-t_{\ast}\right)  }1ds.
\end{equation}
Also, from Thm \ref{Thmtopt*}, we have $y^{\prime}\left(  t_{\ast}\right)
=x_{a}\left(  t_{\ast}\right)  -y\left(  t_{\ast}\right)  =0,$ so a
computation yields%
\begin{equation}
w^{\prime}\left(  t_{\ast}\right)  =2\left[  1+x_{a}\left(  t_{\ast}\right)
-y\left(  t_{\ast}\right)  \right]  \left(  x_{a}^{\prime}\left(  t_{\ast
}\right)  -0\right)  =2x_{a}^{\prime}\left(  t_{\ast}\right)  .
\end{equation}
Using $w\left(  t_{\ast}\right)  =1$, and evaluating the integral, one
obtains
\begin{equation}
Q\left(  t_{\ast}\right)  \leq2x_{a}^{\prime}\left(  t_{\ast}\right)
+\sigma^{2}e^{c\left(  t_{0}-t_{\ast}\right)  }<0.
\end{equation}
The last inequality follows from Condition $E$.
\end{proof}

Hence, recalling that $t_{0}<t_{1}<t_{m}<t_{\ast}$, we obtain the result that
the maximum of $Q,$ the limiting volatility precedes the peak of $y\left(
t\right)  $, which occurs at $t_{\ast}$.

\begin{theorem}
There exists a $t_{v}\in(t_{1},t_{\ast})$ such that $Q^{\prime}\left(
t_{v}\right)  =0.$
\end{theorem}

In summary, the derivative of $y$ catches up to $x_{a}$ at $t_{1}.$ Recalling
$\left(  \ref{Q}\right)  $, we see that $Q\left(  t_{v}\right)  =\sigma
^{-2}d\mathbb{V}\left(  t_{v}\right)  /dt=0$ corresponds to a maximum in $\mathbb{V},$
and this occurs after $t_{1}$ and before $t_{m}$ where $x_{a}$ has its peak.
The peak of $x_{a}$ precedes the peak of $y$ at $t_{\ast}.$ Thus, $\mathbb{V}$
has a maximum prior to the maxima of $x_{a}$ and $y$.

In conclusion, we have shown that the limiting volatility $\mathbb{V}\left(  t\right)
$ attains its maximum prior to that of the expected logarithm of the price,
$y\left(  t\right)$.

\section*{Appendix}

We start with%
\begin{equation}
d\log P\left( t,\omega \right) =G\left( D/S\right) dt+\frac{1}{2}\left\{ 
\frac{D}{S}G^{\prime }\left( \frac{D}{S}\right) +\frac{S}{D}G^{\prime
}\left( \frac{S}{D}\right) \right\} \sigma dW\left( t,\omega \right). \nonumber
\end{equation}
and set $G\left( x\right) :=x-1/x,$ so the model is 
\begin{equation}
d\log P=\left( \frac{D}{S}-\frac{S}{D}\right) dt+\left( \frac{D}{S}+\frac{S}{%
D}\right) \sigma dW  \label{g1}
\end{equation}%
in which the supply and demand are on a symmetric footing. In other words,
when supply exceeds demand, the price moves down in the same way as it moves
up when demand exceeds supply. The coefficient for $dW$ is symmetric in $S$
and $D.$

In order to study market tops and bottoms, we would like to simplify this
expression. We consider the regimes: $\left( i\right) $ $D$ and $S$ deviate
by a small amount, and $\left( ii\right) $ $D$ and $S$ deviate by a large
amount.

$\left( i\right) $ Suppose that $D=q+\delta ^{\prime }$ and $S=q-\varepsilon
^{\prime }$ where $q>0$ and $\delta ^{\prime }$ and $\varepsilon ^{\prime }$
are small in magnitude, i.e., one is not far from equilbrium. Then we have
with $\delta :=\delta ^{\prime }/q$ \ and $\varepsilon :=\varepsilon
^{\prime }/q$ 
\begin{eqnarray*}
\frac{D}{S}-1 &=&\frac{1+\delta }{1-\varepsilon }-1\tilde{=}\delta
+\varepsilon , \\
1-\frac{S}{D} &=&1-\frac{1-\varepsilon }{1+\delta }\tilde{=}\delta
+\varepsilon , \\
\frac{1}{2}\left( \frac{D}{S}-\frac{S}{D}\right)  &=&\frac{1}{2}\left( \frac{%
1+\delta }{1-\varepsilon }-\frac{1-\varepsilon }{1+\delta }\right) \tilde{=}%
\delta +\varepsilon 
\end{eqnarray*}%
So all three terms are equal up through $O\left( \delta ,\varepsilon \right)
.$ Thus, when one is not too far from equilibrium, these terms are
approximately equal and one can use any of them in the deterministic part of
the price equation.

Similarly, under these near equilibrium conditions, the terms $D/S,$ $S/D$
and $(D/S+S/D)/2$ are all equal through $O\left( 1\right) .$

$\left( ii\right) $ Next suppose that we are far from equilibrium, and note
that 
\[
\frac{D}{S}\tilde{=}\frac{D}{S}+\frac{S}{D}\ and\ \ \frac{D}{S}-1\tilde{=}%
\frac{D}{S}\tilde{=}\frac{D}{S}-\frac{S}{D}\ \ if\ \ D/S>>1.
\]%
Similarly, one has 
\[
\frac{S}{D}\tilde{=}\frac{D}{S}+\frac{S}{D}\ and\ \ 1-\frac{S}{D}\tilde{=}-%
\frac{S}{D}\tilde{=}\frac{D}{S}-\frac{S}{D}\ if\ \ S/D>>1.
\]

Applying these approximations to $\left( \ref{g1}\right) $ we see that for
market tops (when $D\geq S$) we can use the model%
\[
d\log P=\left( \frac{D}{S}-1\right) dt+\frac{D}{S}\sigma dW,
\]%
and analogously, for market bottoms, (when $S\geq D$) we use%
\[
d\log P=\left( 1-\frac{S}{D}\right) dt+\frac{S}{D}\sigma dW.
\]%
Note that for the coefficient of $\sigma dW$, we are essentially
approximating $G\left( x\right) :=x+1/x$ with $x$ when $x\geq 1$ and by $1/x$
when $x\leq 1$.  Near $x=1,$ of course, this introduces a factor of $2$ that
can be incorporated into $\sigma .$

\end{document}